\def\draft{0} 
\def\squeezed{0} 
\def\full{1} 
\def\shownotes{1}

\documentclass[11pt]{article}
\usepackage{amsfonts,amsmath,amssymb,color,url,fullpage}

\providecommand{\remove}[1]{}

\providecommand{\ie}{\emph{i.e.}{} }
\providecommand{\eg}{\emph{e.g.}{} }

\providecommand{\Exp}{\mathbb{E}}

\providecommand{\dans}{\rightarrow}

\providecommand{\eps}{\varepsilon}
\providecommand{\zo}{\{0, 1\}}

\renewcommand{\Pr}{\mathop{\mathrm{Pr}}}

\makeatletter
\@addtoreset{equation}{section}
\makeatother

\usepackage[numbers]{natbib} 
\usepackage[latin1]{inputenc}
\usepackage{times}
\usepackage[hylinks,final,titlepage,full,nofullpage,usetoc]{tcs}
\usepackage{hyperref}

\ifnum\shownotes=1
\newcommand{\Dnote}[1]{{\authnote{Dave}{#1}}}
\newcommand{\Snote}[1]{{\authnote{Salil}{#1}}}
\newcommand{\Knote}[1]{{\authnote{Kobbi}{#1}}}
\else
\newcommand{\Dnote}[1]{}
\newcommand{\Snote}[1]{}
\newcommand{\Knote}[1]{}
\fi

\ifnum\squeezed=1
\usepackage{sectsty}
\usepackage[
	hmargin=0.95in,vmargin=1.0in]{geometry}
\usepackage[compact]{titlesec}
\titlespacing{\section}{0pt}{4pt}{4pt}
\titlespacing{\subsection}{0pt}{4pt}{4pt}
\titlespacing{\subsubsection}{0pt}{4pt}{4pt}
\fi

\usepackage{mdwlist}

\newcommand{\Geom}{\mathrm{Geom}}
\newcommand{\Pay}{\mathsf{Pay}}
\newcommand{\Loss}{\mathsf{Loss}}

\newcommand{\Mout}{M_{\mathsf{out}}}
\newcommand{\Mpay}{M_{\mathsf{pay}}}

\title{Redrawing the Boundaries on Purchasing Data from Privacy-Sensitive Individuals}

\author{Kobbi Nissim\thanks{Ben-Gurion University and Harvard
    University. Work done while K.N. was visiting Harvard's Center for
    Research on Computation and Society (CRCS). } \and Salil
  Vadhan\thanks{Center for Research on Computation \& Society and
    School of Engineering \& Applied Sciences, Harvard University,
    Cambridge, MA.  E-mail: \texttt{salil@seas.harvard.edu}.}\and
  David Xiao\thanks{CNRS, Universit\'{e} Paris 7. E-mail: 
    \texttt{dxiao@liafa.univ-paris-diderot.fr}. Part of this work done
    while D.X. was visiting Harvard's Center for 
    Research on Computation and Society (CRCS).}}

\begin{document}

\maketitle

\thispagestyle{empty}
\setcounter{page}{0}

\begin{abstract}

We prove new positive and negative results concerning the existence of truthful and individually rational mechanisms for purchasing private data from individuals with unbounded and sensitive privacy preferences.  
We strengthen the impossibility results of Ghosh and Roth (EC 2011) by extending it to a much wider class of privacy valuations. In particular, these include privacy valuations that are based on $(\eps, \delta)$-differentially private mechanisms for
non-zero $\delta$, ones where the privacy costs are measured in a per-database manner (rather than taking the worst case), and ones that do not depend on the payments made to players (which might not be observable to an adversary).  

\Dnote{changed following sentence} To bypass this impossibility
result, we study a natural special setting where 
individuals have \emph{monotonic privacy valuations}, which captures
common contexts where certain values for private data are expected to
lead to higher valuations for privacy (\eg having a particular
disease).  We give new mechanisms that are individually rational for
all players with monotonic privacy valuations, truthful for all
players whose privacy valuations are not too large, and accurate if
there are not too many players with too-large privacy valuations.  We
also prove matching lower bounds showing that in some respects our
mechanism cannot be improved significantly.  
\Dnote{modified last sentence.  maybe we should remove it altogether?} \Knote{commented last sentence.}
\end{abstract} 

\vfill
\Snote{anything else?} \Knote{consent elicitation?}

\newpage

\section{Introduction}

Computing over individuals' private data is extremely useful for various purposes, such as medical or demographic studies.  
Recent work on {\em differential privacy}~\cite{DMNS06, Dwo06} has focused on ensuring that analyses using private data can be carried out accurately while providing individuals a strong quantitative guarantee of privacy. 

While differential privacy provides formal guarantees on how much
information is leaked about an individual's data, it is silent about
what incentivizes the individuals to share their data in the first place.  A recent line of work \cite{MT07, GR11, NST12, X13, NOS12, CCKMV13, FL12, LR12, RS12} has begun exploring this question, by relating differential privacy to questions of mechanism design.

One way to incentivize individuals to consent to the usage of their private data is simply to pay them for using it.  \Knote{added "to consent ..."}
For example, a medical study may compensate its participants for the use of their medical data.  However, determining the correct price is challenging: low payments may not draw enough
participants, causing insufficient data for an accurate study, while high payments may be impossible for budgetary reasons.

Ghosh and Roth \cite{GR11} approached this problem by allowing the
mechanism to elicit \emph{privacy valuations} from individuals.  A
privacy valuation is a description of how much disutility an
individual experiences from having information about their private data revealed.  By
eliciting valuations, the mechanism is hopefully able to tailor
payments to incentivize enough participants to produce an accurate
result, while not paying too much.

\subsection{The setting and previous work}

We continue the study of purchasing private data from individuals as
first proposed by Ghosh and Roth \cite{GR11} (see \cite{R12,PR13} for a
survey of this area).  Since we work in a game-theoretic framework, we
will also call individuals ``players''.  As in \cite{GR11}, we
study the simple case where the private information consists of a
single data bit, which players can refuse to provide but cannot modify (e.g. because the data is already
certified in a trusted database, such as a medical record database).

To determine the price to pay players for their data bits, the
mechanism elicits \emph{privacy valuations} from them. We study the
simple case where each player $i$'s privacy valuation is parameterized
by a single real parameter $v_i$.  For example, in Ghosh and Roth
\cite{GR11} they assume that player $i$ loses $v_i \eps$ utility when
their data bit is used in an $\eps$-differentially private mechanism.
We will study a wider variety of privacy valuation functions in this
paper.  The valuations are known only to the players themselves, and
therefore players may report false valuations if it increases their
utility.  Furthermore, because these valuations may be correlated with
the data bits, the players may wish to keep their valuations private
as well.  It is instructive to keep in mind the application of paying
for access to medical data (\eg HIV status), where players cannot
control the actual data bit, but their valuation might be strongly
correlated to their data bit.

The goal of the mechanism is to approximate the sum of data bits while
not paying too much.  Based on the declared valuations, the mechanism
computes payments to each of the players and obtains access to the
purchased data bits from the players that accept the payments.
The mechanism then computes and publishes an approximation to the sum of the data bits, which can cause the players some loss of privacy, which should be
compensated for by the mechanism's payment.

The mechanism designer aims to achieve three goals, standard in the
game theory literature: the mechanism should be individually rational,
truthful, and accurate.
A
mechanism is \emph{individually rational} if all players receive
non-negative utility from participating in the game.  In our context, this means
that the mechanism is sufficiently compensating players for their loss in privacy, something
that may be important for ethical reasons, beyond just incentivizing
participation. 
Informally, a mechanism is \emph{truthful} for player $i$
on a tuple $x=(x_1,\ldots,x_n)$ of reports from the players
if player $i$ does not gain in utility by declaring
some false type $x'_i$ (while the other players' types remain
unchanged).  We aim to build mechanisms that are individually rational for all
players, and truthful for as many
players and inputs as possible (ideally for all players and inputs).   A mechanism is
{\em accurate} if the output of the mechanism is close to the true function
it wishes to compute, in our case the sum of the data bits.

Ghosh and Roth \cite{GR11} study the restricted setting (in their
terminology the ``insensitive value model'') where players do not care
about leaking their privacy valuations, as well as the general model
(the ``sensitive value model'') where they may care and their valuations can be unbounded. They present two
mechanisms in the insensitive value model, one that optimizes accuracy
given a fixed budget and another that optimizes budget given a fixed
accuracy constraint.  
They also prove that their
mechanisms are individually rational and truthful under the assumption
that each player $i$ experiences a disutility of \emph{exactly} $v_i
\eps$ when his data bit is used in an $\eps$-differentially private
mechanism.

In the general sensitive value model, they prove the following
impossibility result: there is no individually rational
mechanism with finite payments that can distinguish between
the case where all players have data bit $0$ and the case where all
players have data bit $1$. 

This impossibility result spurred a line of work attempting to bypass
it.  Fleischer and Lyu \cite{FL12} propose a Bayesian setting, where
(for simplicity considering just Boolean inputs) there are
publically known distributions $D_0$ and $D_1$ over privacy valuations, and each player
who has data bit $b_i$ receives a valuation $v_i$ drawn from $D_{b_i}$.  They show
that in this model, it is possible to build a Bayes-Nash truthful,
individually rational, and accurate mechanism.

In a related work, Roth and Schoenebeck \cite{RS12} study a Bayesian
setting where the agents' actual (dis)utilities are drawn from a known
prior, and construct individually rational and ex-post truthful mechanism that are optimal for
minimizing variance given a fixed budget and minimizing expected cost
given a fixed variance goal. In comparison to \cite{FL12}, \cite{RS12}
studies a disutility value that does not quantitively relate to the privacy properties
of the mechanism (but rather just a fixed, per-player disutility for participation), while it results in mechanisms satisfying a
stronger notion of truthfulness.

Ligett and Roth \cite{LR12}
measure the privacy loss incurred
from  a
player's decision to participate
separately from the information leaked
about the
actual data (effectively ruling out arbitrary correlations between privacy valuations and data bits).
They work in a worst-case (non-Bayesian) model and construct a
mechanism that satisfies a relaxed ``one-sided'' notion of
truthfulness and accuracy.  However, their mechanism only satisfies
individual rationality for players whose privacy valuation is not too
high.

\subsubsection{Improving the negative results}

\label{sec:issues}

This line of work leaves several interesting questions open.  The
first is whether the impossibility result of \cite{GR11} really closes the
door on all meaningful mechanisms when players can have unbounded privacy valuations that can be arbitrarily correlated with their sensitive data.

There are two
important loopholes that the 
result 
leaves open.  First, their notion of privacy loss is pure
$\eps$-differential privacy, and they crucially use the fact that for
pure $\eps$-differentially private mechanisms the support of the
output distribution must be identical for all inputs.  This prevents
their result from ruling out notions of privacy loss based on more
relaxed notions of privacy, such as $(\eps, \delta)$-differential
privacy for $\delta > 0$.  As a number of examples in the differential
privacy literature show, relaxing to $(\eps, \delta)$-differential privacy can be
extremely powerful, even when $\delta$ is negligibly small but
non-zero \cite{DworkLe09,HardtTa10,DworkRoVa10,De12,BeimelNiSt13}.
Furthermore, even $(\eps,\delta)$ differential privacy measures the worst-case privacy
loss over all databases, and it may be the case that on most databases, the players' expected privacy loss is much less than the worst case bound.\Knote{modified last sentence to expected loss}%
\footnote{For example, consider a mechanism that computes an $\eps$-differentially private noisy sum of the first $n-1$ rows (which we assume are bits), and if the result is 0, also outputs a $\eps$-differentially private noisy version of the $n$'th row (e.g. via ``randomized response''). 
The worst case privacy loss for player $n$ is $\eps$. 
On databases of the form $(0,0,\ldots,0,b)$ the first computation results with 0 with probability $\approx \eps$ and player $n$ suffers $\eps$ privacy loss with this probability. However, if it is known that the database is very unlikely to be almost entirely zero, then player $n$ may experience any privacy loss with only  exponentially small probability. \label{footnote:perdatabase} \Knote{modified to make more precise. original footnote in file.}}
Thus it is more realistic to use per-database measure of
privacy loss (as done in \cite{CCKMV13}).  

Second, the \cite{GR11} notion of privacy includes as observable and hence potentially disclosive output the (sum
of the) payments made to \emph{all} the players, not just the
sum of the data bits. 
This leaves open the possibility of constructing mechanisms for the setting where an outside observer is
not able to to see some of the player's payments.
For example, it may be natural to assume that, when trying to
learn about player $i$, an observer learns the payments to all players
\emph{except} player $i$.  In the extreme case, we could even restrict
the outside observer to not see any of the payments, but only the approximation to the sum of
the data bits.  The Ghosh-Rosh impossibility proof fails in these
cases. Indeed in this case where player $i$'s own payment is not
visible to the observer, there \emph{does exist} an individually
rational and accurate mechanism with finite payments: simply ask each
player for their valuation $v_i$ and pay them $v_i \eps$, then output
the sum of all the bits with noise of magnitude $O(1/\eps)$.  (The reason
that this mechanism is unsatisfactory is that it is completely untruthful ---
players always gain by reporting a higher valuations.) 

We will close both these gaps: our results will hold even under very
mild conditions on how the players experience privacy loss (in
particular capturing a per-database analogue of $(\eps, \delta)$-differential privacy), and
even when \emph{only} the approximate count of data bits is observable and
\emph{none} of the payments are observable.

\subsubsection{Improving the positive results}

Another question left open by the previous work is whether we can achieve individual
rationality and some form of truthfulness under a worst-case setting.  Recall that \cite{FL12} and
\cite{RS12} work in a Bayesian model, while \cite{LR12} does not guarantee
individual rationality for all players.  Furthermore, in both \cite{FL12} and \cite{RS12}
the priors are heavily used in \emph{designing the mechanism}, and
therefore their results break if the mechanism designer does not
accurately know the priors.
We will replace the Bayesian assumption with a simple
qualitative assumption on the monotonicity of the correlation between players' data bits
and their privacy valuation.  For accuracy (but not individual rationality), we will assume
a rough bound on how many players exceed a given threshold in their privacy
valuations (similarly to \cite{NOS12}).

Another question is the interpretation of the privacy loss functions.
We observe that the truthfulness of the mechanisms in \cite{GR11}
crucially relies on the assumption that $v_i \eps$ is the \emph{exact}
privacy loss incurred.  As was argued by \cite{NOS12} and \cite{CCKMV13},
it seems hard to quantify the exact privacy loss a player experiences, as it may
depend on the mechanism, all of the players' inputs, as well as an adversary's auxiliary
information about the database.  (See \autoref{footnote:perdatabase} for an example.)
It is much more reasonable to assume that the privacy valuations $v_i$
declared by the players and the differential privacy parameter $\eps$ yield an
\emph{upper bound} on their privacy
loss.  When using this interpretation, the truthfulness of \cite{GR11}
no longer holds.  The mechanisms we construct will remain truthful
using the privacy loss function only as an upper
bound on privacy loss (for players whose privacy valuations are not too large,
similarly to the truthfulness guarantees of \cite{NOS12,CCKMV13,LR12}).

\subsection{Our results}
\label{sec:results}

In our model there are $n$ players labelled $1, \ldots, n$ each with a
data bit $b_i \in \zo$ and a privacy valuation $v_i \in \R$, which we describe as
a $2n$-tuple $(b,v)\in \zo^n\times \R^n$.  The
mechanism designer is interested in learning (an approximation of) $\sum b_i$.
The players may lie about their valuation but they cannot lie about
their data bit.  A mechanism $M$ is a pair of randomized functions
$(\Mout, \Mpay)$, where $\Mout : \zo^n \times \R^n \dans \Z$ and
$\Mpay : \zo^n \times \R^n \dans \R^n$.  Namely $\Mout$ produces an
integer that should approximate $\sum b_i$ while $\Mpay$ produces
payments to each of the $n$ players.

Because the players are privacy-aware, the utility they derive from
the game can be separated into two parts as follows:
$$\textrm{utility}_i = \textrm{payment}_i - \textrm{privacy loss}_i.$$ 
(Note that in this paper, we assume the players have no (dis)interest
in the integer that $\Mout$ produces.)  The privacy loss term will be
quantified by a \emph{privacy loss function} that depends on the identity of the player, his bit, his privacy valuation, and his declared valuation $i, b, v, v'_i$ (where $v'_i$ is not necessarily his
true type $v_i$), the mechanism $M$, and the outcome $(s, p)$ produced by $(\Mout,\Mpay)$. \Knote{made verbose.}

\paragraph{Strengthened impossibility result of non-trivial accuracy
  with privacy.}

Our first result significantly strengthens the impossibility result of
Ghosh-Roth \cite{GR11}.

\begin{theorem}[Main impossibility result, informal.  See \autoref{thm:imp}]
  \label{thm:impinf}
  Fix any mechanism $M$ and any reasonable privacy loss functions.
  Then if $M$ is truthful (even if only for players with privacy valuation 0) and
  individually rational and makes finite payments to the players (even if only when
  all players have privacy valuation 0), then
  $M$ cannot distinguish between inputs $(b,v)=(0^n, 0^n)$ and $(b',v)=(1^n, 0^n)$.
\end{theorem}
\Snote{edited parenthetical ``even ifs''}

\Snote{modified next few sentences}
By ``reasonable privacy loss functions,'' we mean that
if from observing the output of the mechanism on an input $(b,v)$, an adversary
can distinguish the case that player $i$ has data bit $b_i=0$ from data bit $b_i=1$ (while keeping
all other inputs the same), then player $i$ experiences a significant privacy loss (proportional to $v_i$) on
database $(b,v)$.
In particular, we allow for a per-database notion of privacy loss.  Moreover,
we only need the adversary to be able to observe the mechanism's estimate of the count $\sum_j b_j$, and not
any of the payments made to players.  And our notion of indistinguishability captures
not
only pure $\eps$-differential privacy but also $(\eps,
\delta)$-differential privacy for $\delta > 0$.
The conclusion of the
result is as strong as conceivably possible, stating that $M$ cannot
distinguish between the two most different inputs (data bits all $0$
vs. data bits all $1$) even in the case where none of the players care
about privacy.

We also remark that in our main impossibility result, in order to
handle privacy loss functions that depend only on the
distribution of the observable count and not the payment information,
we crucially use the requirement that $M$ be truthful for players with
$0$ privacy valuation. As we remarked earlier in
\autoref{sec:issues} there exist $M$ that are individually rational
and accurate (but not truthful).

\paragraph{New notions of privacy and positive results.}

\Dnote{edited this section}
One of the main conceptual contributions of this work is restricting
our attention to a special class of privacy loss functions, which we
use to bypass our main impossibility 
result.  Essential to the definition of differential privacy
(\autoref{def:dp}) is the notion of \emph{neighboring inputs}.  Two
inputs to the mechanism are considered neighboring if they differ only in the
information of a single player, and in the usual notion of
differential privacy, one player's information may differ arbitrarily.
%
This view also characterized how previous work modeled privacy loss
functions: in the sensitive value model of \cite{GR11}, the
privacy loss function to a player $i$ on an input $(b_i, v_i)$ was
computed by considering how much changing to any possible neighbor
$(b'_i, v'_i)$ would affect the output of the mechanism.  In contrast,
we will restrict our attention to privacy loss functions that consider
only how much changing to a specific subset of possible neighbors
$(b'_i, v'_i)$ would affect the output of the mechanism.  By
restricting to such privacy loss functions, we can bypass our
impossibility results.

We now describe how we restrict $(b'_i, v'_i)$.  Recall that in our setting a
single player's type information is a pair $(b_i,v_i)$ where $b_i\in
\zo$ is a data bit and $v_i\in \R$ is a value for privacy.  We observe
that in many cases there is a natural sensitive value of the bit
$b_i$, for example, if $b_i$ represents HIV status, then we would
expect that $b_i = 1$ is more sensitive than $b_i = 0$.

Therefore we consider only the following \emph{monotonic valuations}:
$(0, v_i)$ is a neighbor of $(1, v_i')$ iff $v_i \leq v_i'$.  Thus, if
a player's true type is $(1, v'_i)$, then he is only concerned with how
much the output of the mechanism differs from the case that her actual type were $(0, v_i)$ for $v_i \leq v'_i$.

Consider the pairs that we have excluded: any pairs $(b_i, v_i), (b_i, v'_i)$
(\ie the data bit does not change) and any pairs $(0, v_i), (1, v'_i)$
where $v_i > v'_i$.  By excluding these pairs we formally capture the idea
that players are not concerned about revealing their privacy
valuations \emph{except} inasmuch as they may be correlated with their
data bits $b_i$ and therefore may reveal something about $b_i$.  Since $b_i
= 1$ is more sensitive than $b_i = 0$, the correlation says that privacy
valuation when $b_i = 1$ should be larger than when $b_i = 0$.  This can be seen as an intermediate notion between a model where
players do not care at all about leaking their privacy valuation (the
insensitive value model of \cite{GR11}), and a model where players care
about leaking any and all information about their privacy valuation
(the sensitive value model of \cite{GR11}).

Of course the assumption that players are not concerned about
revealing their privacy valuation except inasmuch as it is correlated
with their data is highly context-dependent.  There may settings where
the privacy valuation is intrinsically sensitive, independently of the
players' data bits, and in these cases using our notion of monotonic
valuations would be inappropriate.  However, we believe that there are
many settings where our relaxation is reasonable.

By using this relaxed notion of privacy, we are able to bypass our
main impossibility result and prove the following:
\begin{theorem}[Main positive result, informal, see \autoref{thm:mon}]
  For any fixed budget $B$ and $\eps > 0$, for privacy loss functions
  that only depend on how the output distribution changes between
  monotonic valuations, there exists a mechanism $M$ that is
  individually rational for all players and truthful for players with low
  privacy valuation (specifically $v_i \leq B/2\eps n$).  Furthermore, as long as the
  players with low privacy valuation do indeed behave truthfully, then
  regardless of the behavior of the players with high privacy
  valuation, the mechanism's output estimates the sum $\sum_i b_i$
  to within $\pm (h + O(1/\eps))$ \Dnote{added parens} where
  $h$ is the number of players with high privacy valuation.
\end{theorem}

Note that even though we fix a budget $B$ beforehand and thus cannot make arbitrarily high payments,
we still achieve individual rationality for all players, even those with extremely high privacy valuations $v_i$.
We do so by ensuring that such players experience perfect privacy ($\eps_i=0$), assuming they have
monotonic valuations.
We
also remark that while we do not achieve truthfulness for all players,
this is not a significant problem as long as the number $h$ of
players with high privacy valuation is not too large.  This is because
the accuracy guarantee holds even if the non-truthful players lie about their valuations.  We also give a small improvement to
our mechanism that ensures truthfulness for all players with data bit $0$, but
at some additional practical inconvenience; we defer the details to
the body of the paper.

\Dnote{edited this paragraph} We remark that besides our specific
restriction to monotonic valuations in this paper, the underlying
principle of studying restricted notions of privacy loss functions by
considering only subsets of neighbors (where the subset should be
chosen appropriately based on the specific context) could turn out to
be a more generally meaningful and powerful technique that is useful
to bypass impossibility results elsewhere in the study of privacy.

\paragraph{Lower bounds on accuracy.}

The above positive result raises the question: can we adaptively
select the budget $B$ in order to achieve accuracy for all inputs,
even those where some players have arbitrarily high privacy
valuations?  Recall that \autoref{thm:impinf} does not
preclude this because we are now only looking at
monotonic valuations, whereas \autoref{thm:impinf} considers
arbitrary valuations.  We nevertheless show that it is impossible:
\begin{theorem}[Impossibility of accuracy for all privacy valuations,
  informal, see \autoref{thm:monimp}]
  For reasonable privacy loss functions that are only sensitive to
  changes in output distribution of monotonic neighbors, any $M$ with
  finite payments that is truthful (even if only on players with $0$
  privacy valuation) and individually rational, there exist player
  privacy valuations $v, v'$ such that $M$ cannot distinguish between
  $(0^n, v)$ and $(1^n, v')$.
\end{theorem}
The exact formal condition on finite payments is somewhat
stronger here than in \autoref{thm:impinf}, but it remains reasonable;
we defer the formal statement to the body of the paper.

Finally, we also prove a trade-off showing that when there
is a limit on the maximum payment the mechanism makes, then accuracy
cannot be improved beyond a certain point, even when considering only
monotonic valuations.  We defer the statement of this result to
\autoref{sec:accurate}.  \Dnote{rewrote this paragraph}

\subsection{Related work}

The relationship between differential privacy and mechanism design was
first explored by \cite{MT07}.  Besides the already mentioned works,
this relationship was explored and extended in a series of works
\cite{NST12}, \cite{X13} (see also \cite{X11priv}), \cite{NOS12},
\cite{CCKMV13} (see also \cite{CCKMV11}), \cite{HK12}, \cite{KPRU12}.
In \cite{MT07,NST12,HK12,KPRU12}, truthfulness refers only to the utility
that players derive from the outcome of the game (as in standard mechanism design)
and differential privacy is treated as a separate property.
The papers
\cite{X13, NOS12, CCKMV13} study
whether and when such mechanisms, which are separately truthful and differentially private,
remain truthful even if the
players are privacy-aware and may incur some loss in utility from the
leakage of the private information.  Differential privacy has also
been used as a technical tool to solve problems that are not
necessarily immediately obvious as being privacy-related; the original
work of \cite{MT07} does this, by using differential privacy to construct
approximately truthful and optimal mechanisms, while more recently,
\cite{KPRU12} use differential privacy as a tool to compute approximate
equilibria.  For more details, we refer the reader to the recent
surveys of \cite{R12, PR13}.

Two ideas we draw on from this literature (particularly \cite{NOS12,CCKMV13})
are (1) the idea that
privacy loss cannot be used as a threat because we
do not know if a player will actually experience the maximal privacy
loss possible, and therefore we should treat privacy loss functions
only as upper bounds on the actual privacy loss, and (2) the idea that
it is meaningful to construct mechanisms that are truthful for players
with reasonable privacy valuations and accurate if most players satisfy this condition.
Our mechanisms are
truthful not for all players but only for players with low privacy
valuation; it will be accurate if the mechanism designer knows enough
about the population to set a budget such that most players have low
privacy valuation (with respect to the budget).

\section{Definitions}

\subsection{Notation}

For two distributions $X, Y$ we let $\Delta(X, Y)$ denote their total
variation distance (i.e. statistical distance).  For an integer $i$ let
$[i] = \{1, \ldots, i\}$.  For any set $S$ and any vector $v \in S^n$,
we let $v_{-i} \in S^{n-1}$ denote the vector $v_1, \ldots, v_{i-1},
v_{i+1}, \ldots, v_n$. We use the following convention: a vector of $n$
entries consisting of $n-1$ variables or constants followed by an
\emph{indexed} variable denotes the vector of $n$ entries with the
last variable inserted at its index. For example \eg $(0^{n-1}, v_i)$
denotes the vector with all zeros except at the $i$'th entry, which contains $v_i$. Some notation about the setting regarding mechanisms etc. was
already introduced in \autoref{sec:results}.

\subsection{Differential privacy}
\Knote{removed "monotonic valuations" from title}
\Snote{moved monotone definitions to positive results section}

\begin{definition}
  Two inputs $(b, v), (b', v') \in \zo^n \times \R^n $ are
  \emph{$i$-neighbors} if $b_j = b'_j$ and $v_j = v'_j$ for all $j
  \neq i$. They are {\em neighbors} if they are $i$-neighbors for some $i
  \in [n]$.
\end{definition}

\begin{definition}
  \label{def:dp}
  A randomized function $f$ is {\em $(\eps, \delta)$-differentially private}
  if for all neighbors $(b, v), (b', v')$, it
  holds that for all subsets $S$ of the range of $f$:
  \begin{equation}
    \label{eq:ahpobinpz}
    \Pr[f(b, v) \in S] \leq e^\eps \Pr[f(b', v') \in S] + \delta.
  \end{equation}
  We say $f$ is $\eps$-differentially private if it is $(\eps,
  0)$-differentially private.
\end{definition}

The symmetric geometric random
variable $\Geom(\eps)$ takes integer values with probability
mass function $\Pr_{x \getsr \Geom(\eps)}[x = k] \propto e^{-\eps
  |k|}$ for all $k \in \Z$.  It is well-known and easy to verify that
for $b \in \zo^n$, the output $\sum b_i + \Geom(\eps)$ is
$\eps$-differentially private.

\subsection{Privacy loss functions}

A {\em privacy loss function} for player $i$ is a real-valued function
$\lambda^{(M)}_i(b, v, v'_i, s, p_{-i})$ taking as inputs the vectors of all
player types $b, v$, player $i$'s declaration $v'_i$ (not necessarily
equal to $v_i$), and a possible outcome $(s, p_{-i}) \in \Z \times
\R^{n-1}$ of $M$. The function also depends on the mechanism $M =
(\Mout, \Mpay)$.  Finally we define
\begin{equation}
\label{eqn:loss}
\Loss^{(M)}_i(b, v, v'_i) =
\Exp_{(s,p) \getsr M(b, (v_{-i}, v'_i))} [\lambda^{(M)}_i(b, v, v'_i,
s, p_{-i})].
\end{equation}

Observe that we have excluded player $i$'s own payment from the
output, as we will assume that an outside observer cannot see player
$i$'s payment.  We let $M_{-i}$ denote the randomized function
$M_{-i}(b, v) = (\Mout(b, v), \Mpay(b, v)_{-i})$.

We comment that, in contrast to \cite{CCKMV13}, we allow $\lambda^{(M)}_i$ to depend on the player's
declaration $v'_i$ to model the possibility that a player's privacy
loss depends on his declaration.  \Snote{I don't follow the next
  sentence.  We don't let the observer see the actual declaration, as
  this would not respect identical output distributions.}
\Knote{suggest to remove next sentence and only say this can only
  strengthen pos results.} \Dnote{i removed the sentence in question}
Allowing this dependence only strengthens our positive
results, while our negative results hold even if we exclude this
dependence on $v'_i$.  
We remark that even if $\lambda^{(M)}_i$ doesn't depend on $v_i'$,  then $\Loss^{(M)}_i$ will still depend on
$v'_i$, since it is an expectation over the output
distribution of $\Mout(b, (v_{-i}, v'_i))$. (See \autoref{eqn:loss}).)

Since the choice of a specific privacy loss function depends heavily
on the context of the mechanism being studied, we avoid fixing a
single privacy loss function and rather study several reasonable
properties that privacy loss functions should have.  Also, while we
typically think of privacy valuation as being positive and privacy
losses as positive, our definition does not exclude the possibility
that players may \emph{want} to lose their privacy, and therefore we
allow privacy valuations (and losses) to be negative.  Our
impossibility results will only assume non-negative privacy loss,
while our constructions handle possibly negative privacy loss
functions as long as the \emph{absolute value} of the privacy loss
function is bounded appropriately.
\Snote{removed  ``thus making our results as strong
as possible in this regard.''}

\subsection{Mechanism design criteria}

\begin{definition}
  A mechanism $M = (\Mout, \Mpay)$ is {\em $([\alpha, \alpha'],
  \beta)$-accurate} on an input $(b, v)\in\zo^n\times\R^n$ if, setting $\overline{b} =
  \frac{1}{n} \sum_{i=1}^n b_i$, it holds that $$\Pr[\Mout(b, v) \notin
  ((\overline{b} - \alpha) n, (\overline{b} + \alpha') n) ] \leq \beta.$$
  We say that $M$ is {\em $(\alpha,\beta)$-accurate} on $(b,v)$ if it
  is $([\alpha,\alpha],\beta)$-accurate. 
\end{definition}
\Snote{changed closed interval to open interval inside probability above.  Else $([1/2,1/2],\beta)$ accuracy
doesn't make sense.  also added def of $(\alpha,\beta)$-accurate.}

We define $\Pay^{(M)}_i(b, v) = \Exp_{p \getsr \Mpay(b, v)}[p_i]$.

\begin{definition}
  Fix $n$, a mechanism $M$ on $n$ players, and privacy loss functions
  $\lambda^{(M)}_1, \ldots, \lambda^{(M)}_n$.  We say $M$ is
  \emph{individually rational} if for all inputs $(b, v) \in \zo^n
  \times \R^n$ and all $i \in [n]$:
  $$\Pay^{(M)}_i(b, v) \geq \Loss^{(M)}_i(b, v, v_i). $$
  $M$ is \emph{truthful for input $(b, v)$ and player $i$} if for
  all $v'_i$ it holds that
  $$\Pay^{(M)}_i(b, v) - \Loss^{(M)}_i(b, v, v_i) \quad
  \geq \quad \Pay^{(M)}(b, (v_{-i}, v'_i)) -
  \Loss^{(M)}_i(b, v, v'_i).$$ %
 $M$ is simply
  \emph{truthful} if it is truthful for all inputs and all players.
\end{definition}

\section{Impossibility of non-trivial accuracy with privacy}
\label{sec:neg}

\Snote{moved distinguishability definition here}
We will use a notion of
distinguishability that captures when a function leaks information
about an input pertaining to a particular player. \Knote{added qualifier}

\begin{definition}
  An input $(b, v)\in \zo^n\times\R^n$ is \emph{$\delta$-distinguishable}
  for player $i$
  with respect to a randomized function $f$ if there is an $i$-neighbor
  $(b', v')$ such that $\Delta(f(b, v),
  f(b', v')) \geq \delta$.
\end{definition}

We choose a notion based on statistical distance because it
allows us to capture $(\eps, \delta)$-differ\-ential privacy even for
$\delta > 0$.  Namely, if there is an input $(b, v)\in \zo^n\times \R^n$ that is $\delta$-distinguishable for
player $i$ with respect to $f$, then $f$ cannot be $(\eps,
\delta')$-differentially private for any $\eps,\delta'$ satisfying $\delta > \delta' + e^\eps -
1 \approx \delta' + \eps$.
However, note that, unlike differential privacy, $\delta$-distinguishability is a {\em per-input} notion, measuring how much privacy loss a player can experience on a particular input $(b,v)$, {\em not} taking the worst case over all inputs. \Snote{this is a new sentence} \Knote{added emphasis.}

For our impossibility result we will require that  any specified privacy loss should be attainable if the
player's privacy valuation is large enough, as long as there is in
fact a noticeable amount of information about the player's type being
leaked (\ie the player's input is somewhat distinguishable).
Note that having unbounded privacy losses is necessary
for having any kind of negative result. If the privacy losses were always upper-bounded by some value $L$, then
a trivially truthful and individually rational mechanism would simply pay every player $L$ and output the exact sum of data bits.
\Snote{new footnote in response to a reviewer comment}\Knote{made Salil's footnote part of text}

\begin{definition}
  \label{def:increasing}
  A privacy loss function $\lambda^{(M)}_i$ for a mechanism $M$ and
  player $i$ is \emph{increasing for $\delta$-distin\-guishability}
  if there exists a real-valued function $T_i$ such that
  for all $\ell > 0$, $b \in \zo^n$ and $v_{-i} \in \R^{n-1}$, if 
  $v_i \geq T_i(\ell, b, v_{-i})$ and if $(b, v)$ is
  $\delta$-distinguishable for player $i$ with respect to $\Mout$, then
  $\Loss^{(M)}_i(b, v, v_i) > \ell$.
\end{definition}

Notice that in our notion of increasing for
$\delta$-distinguishability we only consider distinguishability for
$\Mout$ and not for $(\Mout, \Mpay)$.  Being able to handle this
definition is what makes our impossibility rule out mechanisms even
for privacy loss functions depending only on the distribution of
$\Mout$.

\autoref{def:increasing} implies that the privacy loss functions are unbounded. 
We next define a natural property of loss functions, that for privacy-indifferent players privacy loss is not affected by the particular value reported for $v_i$.

\Snote{moved this def here, and changed terminology from ``centered'' to ``respects indifference''}

\begin{definition}
  A privacy loss function $\lambda^{(M)}_i$ for a mechanism $M$ and
  player $i$ \emph{respects indifference} if
  whenever $v_i = 0$ it follows that $\Loss^{(M)}_i(b, v, v'_i) =
  \Loss^{(M)}_i(b, v, v''_i)$ for all $v'_i, v''_i$.
\end{definition}

\Snote{changed ``finite payments for privacy-insensitive players'' to ``finite payments when all players are privacy-indifferent''}

\Dnote{moved definitions into theorem}

\begin{theorem}
  \label{thm:imp}
  Fix a mechanism $M$ and a number of players $n$, and non-negative
  privacy loss functions $\lambda^{(M)}_1, \ldots, \lambda^{(M)}_n$.
  Suppose that the $\lambda^{(M)}_i$ respect indifference, and are
  increasing for $\delta$-distin\-guishability for some $\delta \leq
  \tfrac{1}{6n}$.

  Suppose that $M$ that satisfies all of the following:
  \begin{itemize*}
  \item $M$ is individually rational.
  \item $M$ has finite payments when all players are
    privacy-indifferent, in the sense that for all $b \in \zo^n$ and
    all $i \in [n]$, it holds that $\Pay^{(M)}_i(b, 0^n)$ is finite.
  \item $M$ is truthful for privacy-indifferent players, namely $M$ is
    truthful for all inputs $(b, v)$ and players $i$ such that $v_i = 0$. 
  \end{itemize*}
  Then it follows that $M$ cannot have non-trivial accuracy in the
  sense that it cannot be $(1/2, 1/3)$-accurate on $(0^n, 0^n)$ and
  $(1^n, 0^n)$.
\end{theorem}
\begin{proof}
  We write $\Pay_i, \Loss_i, \lambda_i$ to denote $\Pay^{(M)}_i,
  \Loss^{(M)}_i, \lambda^{(M)}_i$.  By the assumption that $M$ has
  finite payments when all players are privacy-indifferent, we can define 
	$$P = \max_{i \in [n], b \in \zo^n} \Pay_i(b, 0^n) < \infty.$$  By the
  assumption that all the $\lambda_i$ are increasing for
  $\delta$-indistin\-guishability, we may define a threshold $$L =
  \max_{i \in [n], b \in \zo^n} T_i(P, b, 0^{n-1})$$ such that for all
  $i \in [n], b \in \zo^n, v_i \geq L$, it holds that if $(b,
  (0^{n-1}, v_i))$ is $\delta$-distinguishable, then $\Loss_i(b,
  (0^{n-1}, v_i), v_i) > P$.

\Snote{elaborated description of hybrids}
 
We construct a sequence of $2n+1$ inputs 
  $x^{(1,0)}, x^{(1,1)}, x^{(2,0)}, x^{(2,1)}, \ldots, x^{(n,0)}, x^{(n,1)}, x^{(n+1,0)}$.
  In $x^{(1,0)}$, all players have data bit 0 and privacy valuation 0.  That is, $x^{(1,0)}=(0^n,0^n)$. From $x^{(i,0)}$, we construct $x^{(i,1)}$ by changing player $i$'s data bit $b_i$ from 0 to 1 and valuation $v_i$ from 0 to $L$.
  From $x^{(i,1)}$, we construct $x^{(i+1,0)}$ by changing player $i$'s valuation $v_i$ back from $L$ to $0$ (but $b_i$ remains 1).
  Thus, 
  \begin{eqnarray*}
	x^{(i,0)} & = &  ((1^{i-1},  0,  0^{n-i}), (0^{i-1}, 0, 0^{n-i})), \quad\textrm{and} \\
	x^{(i, 1)} & = & ((1^{i-1},  1,  0^{n-i}), (0^{i-1}, L, 0^{n-i}))
  \end{eqnarray*}
	In particular, $x^{(n+1,0)} = (1^n, 0^n)$. \Knote{tried to make hybrids visually clear}
  Define the hybrid distributions $H^{(i,j)} = \Mout(x^{(i,j)})$. 

\Snote{added claim and reordered proof a bit}
\begin{claim} For all $i\in [n]$, $\Pay_i(x^{(i,1)}) \leq
  \Pay_i(x^{(i+1,0)}) \leq P$.\end{claim}
  
  To prove this claim, we first note that all players have privacy valuation 0 in $x^{(i+1,0)}$, so $\Pay_i(x^{(i+1,0)}) \leq P$ by 
  the definition of $P$.
  Since player $i$ has privacy valuation 0 in $x^{(i+1,0)}$, we also know that
  privacy loss of player $i$ in input $x^{(i+1,0)}$ is independent of
  her declaration (since $\lambda_i$ respects indifference).  If player $i$ declares $L$ as her valuation instead of 0, she would
  get payment $\Pay_i(x^{(i,1)})$.  By
  truthfulness for privacy-indifferent players, we must have
  $\Pay_i(x^{(i,1)}) \leq\Pay_i(x^{(i+1,0)})$.

  By the definition of $L$ it follows that $x^{(i,1)}$ cannot be
  $\delta$-distinguishable for player $i$ with respect to $\Mout$.
  Otherwise, this would contradict individual rationality because on
  input $x^{(i,1)}$ player $i$ would have privacy loss $> P$ while
  only getting payoff $\leq P$.

  Since $x^{(i,1)}$ is not $\delta$-distinguishable for player $i$
  with respect to $\Mout$, and because $x^{(i,1)}$ is an $i$-neighbor of
  $x^{(i, 0)}$ as well as $x^{(i+1, 0)}$, it follows that
  \begin{equation}
    \label{eq:ahgpoiha}
    \Delta(H^{(i,0)}, H^{(i,1)}) < \delta \text{ and
    }\Delta(H^{(i,1)}, H^{(i+1,0)}) < \delta
  \end{equation}
  Finally, since \autoref{eq:ahgpoiha} holds for all $i \in [n]$, and
  since $H^{(1,0)} = \Mout(0^n, 0^n)$ and $H^{(n+1,0)} = \Mout(1^n, 0^n)$,
  we have by the triangle inequality that
  $$\Delta(\Mout(0^n, 0^n), \Mout(1^n, 0^n)) < 2n\delta$$
  But since $\delta \leq 1/6n$, this contradicts the fact that
  $M$ has non-trivial accuracy, since non-trivial accuracy implies
  that we can distinguish between the output of $\Mout$ on inputs
  $(0^n, 0^n)$ and $(1^n, 0^n)$ with advantage greater than $1/3$,
  simply by checking whether the output is greater than $n/2$.
\end{proof}

\subsection{Subsampling for low-distinguishability privacy loss
  functions}

We comment that the $\delta\leq 1/6n$ bound in \autoref{thm:imp} is tight up to a constant factor.
Indeed, if players do not incur significant losses when their
inputs are $O(1/n)$-distinguishable, then an extremely
simple mechanisms based on sub-sampling can be used to achieve
truthfulness, individual rationality, and accuracy with finite budget.

Namely, suppose that the privacy loss functions are such that if for
all $i$, if player $i$'s input is not $C/n$-distinguishable
for some constant $C$, then regardless of $v_i$, the loss to player
$i$ is bounded by $P$.  Then the following mechanism is truthful,
individually rational, and accurate: pay all players $P$, select at
random a subset $A$ of size $k$ for some $k<C$ from the population, and
output $(n/|A|)\cdot \sum_{i \in A} b_i$. \Snote{added a factor of $n$ to output} 
By a Chernoff Bound, this
mechanism is $(\eta, 2e^{-\eta^2k})$-accurate for all $\eta>0$.  By construction no
player's input is $C/n$-distinguishable and therefore their
privacy loss is at most $P$ and the mechanism is individually
rational.  Finally mechanism is truthful since it behaves
independently of the player declarations.

\section{Positive results}

\Dnote{moved this part here since it's not really about monotonic
  valuations}
For our positive results, we will require the following
natural property from our privacy loss functions.  Recall that we
allow the privacy loss functions $\lambda^{(M)}_i$ to depend on a
player's report $v'_i$, in addition the the player's true type.  We
require the dependence on $v'_i$ to be well-behaved in that if
changing declarations does not change the output distribution, then it
also does not change the privacy loss.
\begin{definition}
  A privacy loss function $\lambda^{(M)}_i$ \emph{respects identical
    output distributions} if the following holds: for all $b, v$, if
  the distribution of $M_{-i}(b, v)$ is identical to $M_{-i}(b, (v_{-i},
  v'_i))$, then for all $s,p$, it holds that $\lambda^{(M)}_i(b, v,
  v'_i, s, p) = \lambda^{(M)}_i(b, v, v_i, s, p)$.
\end{definition}
\Snote{changed $\Mout$ to $M_{-i}$ in the above definition, since it seems to be what we want}
The above definition captures the idea that if what the privacy adversary can see
(namely the output of $M_{-i}$) doesn't change, then player $i$'s privacy loss should not change.

\subsection{Monotonic valuations}

\Dnote{added a sentence}
We now define our main conceptual restriction of the privacy loss
functions to consider only \emph{monotonic valuations}.

\Snote{moved defs of monotone neighbors, monotone dp here}

\Snote{redefined monotonically related to make symmetry explicit}
\begin{definition}
  Two player types $(b_i, v_i), (b'_i, v'_i)\in \zo\times \R$ are said to be
  \emph{monotonically related} iff ($b_i=0$, $b'_i=1$, and $v_i \leq
  v'_i$) or ($b_i=1$, $b'_i=0$, and $v_i \geq v'_i$). \Dnote{added
    subscripts to variables in previous sentence to stay consistent
    with convention that $b, v$ denote vectors}
    Two inputs $(b, v),
  (b', v')\in \zo^n\times \R^n$ are \emph{monotonic $i$-neighbors} if they are
  $i$-neighbors and furthermore $(b_i, v_i), (b'_i, v'_i)$ are
  monotonically related.  They are {\em monotonic neighbors} if they are
  monotonic $i$-neighbors for some $i \in [n]$.
\end{definition}

\Dnote{moved monotonic dp definition, moved definition of bounded by
monotonic dp here}

\Snote{new sentence}
Following \cite{CCKMV13}, we also make the assumption that the privacy loss functions on a given output $(s,p_{-i})$ are
bounded by the amount of influence that player $i$'s report has on the probability of the output:

\begin{definition}
  A privacy loss function $\lambda^{(M)}_i$ is \emph{bounded by
    differential privacy} if the following holds:
  $$ \left|\lambda^{(M)}_i(b, v, v'_i, s, p_{-i})\right| \leq v_i \cdot \left(
    \max_{(b''_i, v''_i)} \log \frac{\Pr[M_{-i}(b, v) = (s, p_{-i})]}
    {\Pr[M_{-i}((b_{-i}, b''_i), (v_{-i}, v''_i)) = (s, p_{-i})]} \right)$$

  A privacy loss function $\lambda^{(M)}_i$ is \emph{bounded by
    differential privacy for monotonic valuations} if:
  $$ \left|\lambda^{(M)}_i(b, v, v'_i, s, p_{-i}) \right| \leq v_i
  \cdot \left( \max_{(b''_i, v''_i) \text{ mon. related to }
      (b_i, v_i)} \log \frac{\Pr[M_{-i}(b, v) = (s, p_{-i})]}
    {\Pr[M_{-i}((b_{-i}, b''_i), (v_{-i}, v''_i)) = (s, p_{-i})]}
  \right) $$
\end{definition}
\Snote{worth thinking about whether there is a Bayesian interpretation of the monotone version of the definition} 

As noted and used in \cite{CCKMV13}, the RHS in the above definition can be upper-bounded by the level of
(pure) differential privacy, and the same holds for monotonic valuations:
\Snote{dropped mention of $(\eps,\delta)$ - seemed to be a distraction}

\begin{fact}
  \label{fact:bound}
  If $M_{-i}$ is $\eps$-differentially private for $i$-neighbors (\ie
  \autoref{eq:ahpobinpz} holds for all $i$-neighbors) and
  $\lambda^{(M)}_i$ is bounded by differential privacy (even if only
  for monotonic valuations), then player $i$'s privacy loss is bounded
  by $v_i \eps$ regardless of other player types, player declarations,
  or outcomes.
\end{fact}

\Dnote{added this subsection}
\Snote{removed subsection heading and made minor edits below}

As hinted at in the definition of privacy loss functions bounded by
differential privacy for monotonic valuations, one can define
an analogue of differential privacy where we take the maximum over just monotonically
related neighbors.  However this notion is not that different from the original notion of differential privacy, \Knote{change "not that interesting" to "not that different from differential privacy"} since satisfying such a definition for
some privacy parameter $\eps$ immediately implies satisfying
(standard) differential privacy for privacy parameter $3\eps$, since
every two pairs $(b_i,v_i)$ and $(b_i',v_i')$ are at distance at most 3 in the
monotonic-neighbor graph.  \Knote{should the factor be 3 or 2?} The monotonic neighbor notion becomes more interesting
if we consider a further variant of differential privacy where
the privacy guarantee $\eps_i$ afforded to an individual depends on her data $(b_i,v_i)$
(e.g. $\eps_i=1/v_i$).  We defer exploration of this notion to a future version of this paper.  

\Dnote{should we define the other definition with individual $\eps_i$ here?}

\subsection{Mechanism for monotonic valuations}

The idea behind our mechanism for players with monotic valuations (\autoref{alg:mon})
is simply to treat the data bit as 0 (the insensitive value) for all players who value privacy too much.

\begin{algorithmf}{Mechanism for monotonic valuations \label{alg:mon}}
  Input: $(b, v) \in \zo^n \times \R^n$. Auxiliary inputs: budget $B
  > 0$, privacy parameter $\eps > 0$.
\begin{enumerate}
\item For all $i \in [n]$, set $b'_i = b_i$ if $2\eps v_i \leq
  B/n$, otherwise set $b'_i = 0$.
\item Output $\sum_{i=1}^n b'_i + \Geom(\eps)$.
\item Pay $B/n$ to player $i$ if $2\eps v_i \leq B/n$, else pay
  player $i$ nothing.
\end{enumerate}
\end{algorithmf}

\begin{theorem}
  \label{thm:mon}
  For privacy loss functions that are bounded by differential privacy
  for monotonic valuations and respect identical output distributions,
  the mechanism $M$ in \autoref{alg:mon} satisfies the following:
  \begin{enumerate}
  \item \label{item:truth} $M$ is truthful for all players with $2\eps
    v_i \leq B/n$.
  \item $M$ is individually rational for all players
  \item Assume only that the truthful players described in Point
    \ref{item:truth} do indeed declare their true types. Letting
    $\eta$ denote the fraction of players where $b_i = 1$ and $2\eps
    v_i > B / n$, it holds that $M$ is $([\eta + \gamma, \gamma],
    2e^{-\eps \gamma n})$-accurate.
  \end{enumerate}
\end{theorem}
\begin{proof}

  \textbf{Truthfulness for players with $2\eps v_i \leq B/n$}: if
  $2\eps v_i \leq B/n$, then declaring any $v'_i \leq B/(2\eps
    n)$ has no effect on the output of the mechanism, and so there is
  no change in utility since the privacy loss functions respect
  identical output distributions.  If player $i$ declares some $v'_i >
  B/(2 \eps n)$, then he loses $B/n$ in payment.  Because
  $M_{-i}$ is $\eps$-differentially private if for $i$-neighbors
  (recall we assume an observer cannot see the change in $p_i$) and we
  assumed that the privacy loss functions are bounded by differential
  privacy for monotonic valuations,  it follows that player $i$'s privacy
  loss has absolute value at most $2\eps v_i$ under a report of $v_i$ and
  under a report of $v_i'$  (\autoref{fact:bound}).
Thus, there is at most a
  change of $2 \eps v_i$ in privacy, which is not sufficient
  to overcome the payment loss of $B/n$. \Snote{elaborated last couple of sentences}

\Snote{reordered proof of IR below}
  \textbf{Individual rationality:} consider any vector of types $b,
  v$ and any player $i$.  If $v_i \leq B/2\eps n$ then player
  $i$ receives payment $B/n$.  By the hypothesis that the privacy loss
  functions are bounded by differential privacy for monotonic
  valuations, and because the mechanism is $\eps$-differentially
  private, the privacy loss to player $i$ is bounded by $\eps v_i <
  B/n$ (\autoref{fact:bound}), satisfying individual rationality.

  Now suppose that player $i$ has valuation $v_i > \tfrac{B}{2 \eps
    n}$.  In this case the payment is $0$.
  The mechanism sets $b'_i = 0$, and for every $(b''_i, v''_i)$
  monotonically related to $(b_i, v_i)$ the mechanism also sets $b'_i
  = 0$.  Since the report of player $i$ does not affect $b'_j$ or
  the payment to player $j$ for $j\neq i$,
 monotonic neighbors will produce the \emph{exact} same output
  distribution of $M_{-i}$.

  Therefore the privacy loss of player $i$ is 0.  Indeed,
  since the
  privacy loss function is bounded by differential privacy for
  monotonic valuations, we have:
  $$ \left|\lambda^{(M)}_i(b, v, v'_i, s, p_{-i}) \right| \leq v_i
  \cdot \left( \max_{(b''_i, v''_i) \text{ mon. related to } (b_i,
      v_i)} \log \frac{\Pr[M_{-i}(b, v) = (s, p_{-i})]}
    {\Pr[M_{-i}((b_{-i}, b''_i), (v_{-i}, v''_i)) = (s, p_{-i})]}
  \right) = 0$$ %

  \textbf{Accuracy:} the bits of the $(1 - \eta)$ fraction of truthful
  players and players with $b_i = 0$ are always counted correctly,
  while the bits of the $\eta$ fraction of players with $b_i = 1$ and
  large privacy valuation $v_i \geq B/(2 \eps n)$ are either
  counted correctly (if they declare a value less than $B/(2
    \eps n)$) or are counted as $0$ (if they declare otherwise).

  This means that $\overline{b'} = \sum_{i=1}^n b'_i$ and
  $\overline{b} = \sum_{i=1}^n b_i$ satisfy
  $\overline{b'}\in [\overline{b}-\eta n,\overline{b'}]$.
  \Snote{removed $1/n$ factors}

  By the definition of symmetric geometric noise, it follows that
  (letting $v'$ be the declared valuations of the players) it holds
  that
  $$\Pr[|\Mout(b, v') - \overline{b'}| \geq \gamma n] < 2
  e^{-\eps \gamma n}.$$
  \Snote{changed strict inequalities to non-strict in this tail bound
    to match new def of accuracy - please check} \Dnote{fine with me -
    I made the bound itself strict to stay consistent with the
    definition, this is true from the definition of geometric noise}
  The theorem follows.

\end{proof}

\subsubsection{Achieving better truthfulness}

We can improve the truthfulness of \autoref{thm:mon} to include all
players with data bit $0$.

\begin{theorem}
  \label{thm:moretruth}
  Let $M'$ be the same as in \autoref{alg:mon}, except that \emph{all}
  players with $b_i = 0$ are paid $B/n$, even those with large privacy
  valuations. Suppose that the $\lambda^{(M')}_i$ are bounded by
  differential privacy for monotonic valuations and also respect identical
  output distributions.  Then the conclusions of \autoref{thm:mon}
  hold and in addition the mechanism is truthful for all players with
  data bit $b_i = 0$.
\end{theorem}

Note that, unlike \autoref{alg:mon}, here the payment that the mechanism makes to players
depends on their data bit, and not just on their reported valuation.  This might make it impractical
in some settings (e.g. if payment is needed before players give permission to view their data bits).

\begin{proof}
  Increasing the payments to the players with $b_i = 0$ and privacy
  valuation $v_i > \frac{B}{2 \eps n}$ does not hurt individual
  rationality or accuracy.  We must however verify that we have not
  harmed truthfulness.  Since players are not allowed to lie about
  their data bit, the same argument for truthfulness of players with
  $b_i = 1$ and $v_i \leq B/(2 \eps n)$ remains valid.  It is
  only necessary to verify that truthfulness holds for all players with
  $b_i = 0$.

  Observe that for players with $b_i = 0$, the output distribution of
  the mechanism is identical regardless of their declaration for
  $v_i$.  Therefore by the assumption that the $\lambda^{(M)}_i$
  respect identical output distributions, changing their declaration
  does not change their privacy loss.  Furthermore, by the definition
  of $M'$ changing their declaration does not change their payment as
  all players with $b_i = 0$ are paid $B/n$.  Therefore, there is no
  advantage to declaring a false valuation.
\end{proof}

We remark that \autoref{thm:moretruth} is only preferable to \autoref{thm:mon} in settings
where knowing the true valuations has some value beyond simply helping
to achieve an accurate output; in particular, notice that $M'$ as
defined in \autoref{thm:moretruth} does not guarantee any better
accuracy or any lower payments (indeed, it may make more payments than
the original \autoref{alg:mon}).

\section{Lower bounds}

\subsection{Impossibility of non-trivial accuracy for all privacy
  valuations with monotonic privacy}

One natural question that \autoref{alg:mon} raises is whether we can
hope to adaptively set the budget $B$ based on the valuations of the players 
and thereby achieve accuracy
for all inputs, not just inputs where most players' privacy valuations
are small relative to some predetermined budget.  In this section we
show that this is not possible, even when only considering players who
care about privacy for monotonic neighbors.

\Snote{moved monotone distinguishability def here}
\begin{definition}
  An input $(b, v)\in \zo^n\times \R^n$ is \emph{$\delta$-monotonically distinguishable} for player $i$
  with respect to a randomized function $f$ if there is a monotonic $i$-neighbor $(b', v')$ such that $\Delta(f(b, v),
  f(b', v')) \geq \delta$.
\end{definition}

\begin{definition}
  \label{def:increasingmon}
  A privacy loss function $\lambda^{(M)}_i$ for a mechanism $M$ and
  player $i$ is  \emph{increasing for $\delta$-monotonic distinguishability}
  if there exists a real-valued function $T_i$ such that
  for all $\ell > 0$, $b \in \zo^n$ and $v_{-i} \in \R^{n-1}$, if 
  $v_i \geq T_i(\ell, b, v_{-i})$ and if $(b, v)$ is
  $\delta$-monotonically
  distinguishable for player $i$ with respect to $\Mout$, then
  $\Loss^{(M)}_i(b, v, v_i) > \ell$.
\end{definition}

\Snote{changed 
``has finite payments for privacy-aware players'' to ``always has finite payments''}
\Dnote{merged definitions into theorem statement.}

\begin{theorem}
  \label{thm:monimp}
  Fix a mechanism $M$ and a number of players $n$, and non-negative
  privacy loss functions $\lambda^{(M)}_1, \ldots, \lambda^{(M)}_n$.
  Suppose that the $\lambda^{(M)}_i$ respect indifference and are
  increasing for $\delta$-monotonic distinguishability for $\delta
  \leq \tfrac{1}{3n}$.

  Suppose $M$ satisfies all the following:
  \begin{itemize*}
  \item $M$ is individually rational.
  \item $M$ always has finite payments, in the sense that  for all
    $b \in \zo^n, v \in \R^n$ and all $i \in [n]$ it holds that
    $\Pay^{(M)}_i(b, v)$ is finite.
  \item $M$ is truthful for privacy-indifferent players, as in
    \autoref{thm:imp}
  \end{itemize*}
  Then $M$ does not have non-trivial accuracy for all privacy
  valuations, namely $M$ cannot be $(1/2, 1/3)$-accurate on $(0^n, v)$
  and $(1^n, v)$ for all $v \in \R^n$.
\end{theorem}
\begin{proof}
  The argument follows the same outline as the proof of
  \autoref{thm:imp}, \ie by constructing a sequence of hybrid inputs
  and using truthfulness for privacy-indifferent players and
  individual rationality to argue that the neighboring hybrids must
  produce statistically close outputs.  However, we have to take more
  care here because for the hybrids in this proof there is no uniform
  way to set the maximum payment $P$ and threshold valuation $L$ for
  achieving privacy loss $>P$ at the beginning of the argument, since here we
  allow the finite payment bound to depend on the valuations (whereas \autoref{thm:imp}
  only refers to the payment bound when all valuations are zero).  Instead, we
  set $P_i, L_i$ for the $i$'th hybrids in a way that depends on $L_{[i-1]} = 
  (L_1,\ldots,L_{i-1})$.

\Snote{elaborated hybrids, and introduced double-indexed hybrid inputs like \autoref{thm:imp}}

  As before, we have $2n+1$ inputs
  $x^{(1,0)}, x^{(1,1)}, x^{(2,0)}, x^{(2,1)}, \ldots, x^{(n,0)}, x^{(n,1)}, x^{(n+1,0)}$, which we define inductively as follows.
In $x^{(1,0)}$, all players have data bit 0 and privacy valuation 0.  That is, $x^{(1,0)}=(0^n,0^n)$.
 From $x^{(i,0)}$, we define $x^{(i,1)}$ by changing player $i$'s data bit from 0 to 1. 
  From $x^{(i,1)}=(b^{(i)},v^{(i)})$, we define $P_i = \Pay_i(x^{(i,1)})$ to be the amount that player $i$ is
  paid in $x^{(i,1)}$, and $L_i = T_i(P_i,b^{(i)},v^{(i)}_{-i})$ to be a privacy valuation beyond which payment $P_i$ does not
  compensate for $\delta$-distinguishability (as promised by \autoref{def:increasingmon}).  Then we define $x^{(i+1,0)}$ by increasing the valuation of player $i$ from 0 to $L_i$. 
  By induction, for $i = 1, \ldots, n+1$, we have $$x^{(i,0)} = (1^{i-1} 0^{n-i+1}, L_{[i-1]} 0^{n-i+1}).$$  \Dnote{fixed indices in previous sentece,
    was off by one} Define the distribution $H^{(i)} = \Mout(x^{(i,0)})$.

  \begin{claim}
    \label{clm:paogih}
    $\Pay_i(x^{(i+1,0)}) \leq \Pay_i(x^{(i,1)}) = P_i$
  \end{claim}
 
  On input $x^{(i,1)}$, player $i$ has privacy valuation 0, so his privacy loss 
  is independent of his declaration (since $\lambda_i$ respects indifference).  Declaring 
  $L_i$ would change the input to $x^{(i+1,0)}$, so by truthfulness for privacy-indifferent players,
  we have $\Pay_i(x^{(i+1,0)}) \leq \Pay_i(x^{(i,1)})$.

  By the definition of $L_i$, $x^{(i+1,0)}$ cannot be
  $\delta$-monotonically distinguishable for player $i$ with respect
  to $\Mout$.  Otherwise, this would contradict individual rationality
  because on input $x^{(i+1,0)}$ player $i$ would have privacy loss greater than $P_i$
  while only getting a payoff of at most $P_i$ (by \autoref{clm:paogih}).

  Since $x^{(i+1,0)}$ is not $\delta$-monotonically distinguishable for
  player $i$ with respect to $\Mout$, and because $x^{(i,0)}$ is an
  $i$-monotonic neighbor of $x^{(i+1,0)}$, it follows that
  $\Delta(H^{(i-1)}, H^{(i)}) < \delta$.
  Finally, since this holds for all $i \in [n]$, the triangle
  inequality implies that $\Delta(H^{(0)}, H^{(n)}) < n \delta$.
  But since $\delta \leq 1/3n$, this implies that
  $$\Delta(\Mout(0^n, 0^n), \Mout(1^n, L)) < 1/3$$
  contradicting the fact that $M$ has non-trivial accuracy for
  all privacy valuations.
\end{proof}

\subsection{Tradeoff between payments and accuracy}
\label{sec:accurate}
\Snote{as discussed, I think we need to change the interpretation of this result.  also the presentation can probably use
editing like done in the other lower bounds} \Dnote{I rewrote the
proof along the lines of the previous ones}

One could also ask whether the accuracy of \autoref{thm:mon} can be
improved, \ie whether it is possible to beat $(\eta + \gamma,
2e^{-\eps \gamma n})$-accuracy.  We now present a result that,
assuming the mechanism does not exceed a certain amount of payment,
limits the best accuracy it can achieve.  (We note however that this
bound is loose and does not match our mechanism.)

In order to prove optimality we will require that the privacy loss
functions be growing with statistical distance, which is a strictly
stronger condition than being increasing for
$\delta$-distinguishability.  However, a stronger requirement is
unavoidable since one can invent contrived privacy loss functions that
are increasing but for which one can achieve $(\eta, 0)$-accuracy by
simply by outputting $\sum b'_i$ as constructed in \autoref{alg:mon}
without noise (while preserving the same truthfulness and individual
rationality guarantees).  Nevertheless, being growing with statistical
distance for monotonic neighbors is compatible with being bounded by
differential privacy for monotonic neighbors (\ie there exist
functions that satisfy both properties), and therefore the following
result still implies limits to how much one can improve the accuracy
of \autoref{thm:mon} for all privacy loss functions bounded by
differential privacy for monotonic neighbors.

\Snote{I would call this ``grows with statistical distance''. It's not
  proportional to, since it's a one-sided inequality} \Dnote{I changed
  the terminology}
\begin{definition}
  $\lambda^{(M)}_i(b, v, v'_i, s, p_{-i})$ is \emph{growing with
    statistical distance (for monotonic neighbors)} if:
  $$\Loss^{(M)}_i(b, v, v_i) \geq v_i \cdot \left(
    \max_{(b', v')} \Delta(\Mout(b, v), \Mout(b', v')) \right)$$ %
  where the maximum is taken over $(b', v')$ that are (monotonic)
  $i$-neighbors of $(b, v)$.
\end{definition}

\begin{theorem}
  Fix a mechanism $M$, a number of players $n$, and privacy loss
  functions $\lambda^{(M)}_i$ for $i = 1, \ldots, n$.  Suppose that
  the $\lambda^{(M)}_i$ respect indifference and are growing with
  statistical distance for monotonic neighbors. \Dnote{yes for
    monotonic neighbors}

  Suppose that $M$ satisfies the following:
  \begin{itemize*}
  \item $M$ is individually rational.
  \item There exists a maximum payment over all possible inputs that
    $M$ makes to any player who declares $0$ privacy valuation.
    Call this maximum value $P$.
  \item $M$ is truthful for privacy-indifferent players as defined
    in \autoref{thm:imp}.
  \end{itemize*}
  Then it holds that for any $\tau, \gamma, \eta > 0$ such that $\eta
  + 2\gamma \leq 1$, and any $\beta < \frac{1}{2} - \tfrac{P}{\tau}
  \gamma n$, the mechanism $M$ cannot be $([\eta + \gamma, \gamma],
  \beta)$-accurate on all inputs where at most an $\eta$ fraction of
  the players' valuations exceed $\tau$.
\end{theorem}
\begin{proof}
  Fix any $\tau, \eta, \gamma > 0$ and any $\beta < \tfrac{1}{2} -
  \tfrac{P}{\tau} \gamma n$.  We prove the theorem by showing that $M$
  cannot be $([\eta + \gamma, \gamma], \beta)$-accurate.  Let $h =
  \eta n$ denote the number of players with high privacy valuation
  allowed.

  Fix any $L \geq P h / (1 - 2 \tfrac{P}{\tau} \gamma n -
  2\beta)$.  Consider the following sequence of hybrid inputs.  Let
  $x^{(1,0)} = (0^n, 0^n)$.  From $x^{(i,0)}$, define $x^{(i,1)}$ by
  flipping player $i$'s data bit from $0$ to $1$.  From $x^{(i,1)}$,
  define $x^{(i+1,0)}$ by increasing the valuation of player $i$ from
  $0$ to $L$ if $i \in [h+1]$, or from $0$ to $\tau$ if $i \in (h+1, h + 2
  \gamma n+1]$.  By induction, we have:
  \begin{eqnarray*}
    \forall i \in [h+1], & x^{(i,0)} = & (1^{i-1} 0^{n-i+1}, L^{i-1}
    0^{n-i+1}) \\
    \forall i \in (h+1, h+2\gamma n+1], & x^{(i, 0)} = & (1^{i-1}
    0^{n-i+1}, L^h \tau^{i-h-1} 0^{n-i+1}) 
  \end{eqnarray*}
  These are well-defined since $h+2\gamma = (\eta + 2\gamma) n \leq
  n$.  Define the hybrids $H^{(i,0)} = \Mout(x^{(i,0)})$.  To analyze
  these hybrids, we use the following claims.

  \begin{claim}
    \label{claim:indist}
    For any input $(b, v)$ where player $i$ is paid at most $P$, it
    holds that $(b, v)$ is not $\delta$-distinguishable for monotonic
    neighbors for player $i$ with respect to $\Mout$ for any $\delta
    \geq P / v_i$.
  \end{claim}

  \autoref{claim:indist} holds because by individual rationality, it
  holds that the privacy loss does not exceed $P$.  By the assumption
  that the privacy loss functions are growing with statistical
  distance for monotonic neighbors, it follows that $\Delta(\Mout(b,
  v), \Mout(b', v')) \leq P / v_i$ for all $(b', v')$ monotonic
  neighbors of $(b, v)$.

  \begin{claim}
    \label{claim:apobnipbn}
    $\Pay_i(x^{(i+1,0)}) \leq \Pay_i(x^{(i,1)}) \leq P$.
  \end{claim}
  As in the proof of \autoref{thm:monimp}, this claim holds because on
  input $x^{(i,1)}$, player $i+1$ has $0$ privacy valuation, and so
  $\Pay_i(x^{(i,1)}) \leq P$ by our assumption that the mechanism pays
  at most $P$ to players with $0$ privacy valuation.
  $\Pay_i(x^{(i+1,0)}) \leq \Pay_i(x^{(i,1)})$ follows as in the proof of
  \autoref{thm:monimp} from the truthfulness of the mechanism for
  privacy-indifferent players and by the fact that the privacy loss
  functions respect indifference.

  We may apply \autoref{claim:indist} to conclude that for all $i \in
  [h]$, since player $i$ has valuation $L$ in $x^{(i+1,0)}$, it holds
  that $x^{(i+1,0)}$ cannot be $(P/L)$-distinguishable for monotonic
  neighbors for player $i$.  Since $x^{(i,0)}, x^{(i+1,0)}$ are
  monotonic $i$-neighbors, it follows that $\Delta(H^{(i,0)}, H^{(i+1,0)})
  < P/L$.

  Repeating the same argument for all $i \in [h+1, h +2\gamma n]$
  and using the fact that player $i$ has valuation $\tau$ in
  $x^{(i+1,0)}$ for these $i$, it follows that $\Delta(H^{(i,0)},
  H^{(i+1,0)}) < P / \tau$.

  Combining the above using the triangle inequality and applying the
  definition of $L$, we deduce that
  \begin{equation}
    \label{eq:ahpgoiahsg}
    \Delta(H^{(1,0)}, H^{(h+ 2\gamma n+1, 0)}) < \frac{\eta n P}{L} +
    \frac{2\gamma n P}{\tau} \leq 1 - 2\beta
  \end{equation}

  For $i \in [n]$, define the open interval on the real line $A(i) = (i-1
  -(\eta + \gamma)n, i-1 + \gamma n)$.  Since the sum of the data bits
  in $x^{(i,0)}$ is exactly $i-1$, in order for $M$ to be $([\eta +
  \gamma, \gamma], \beta)$-accurate, it is necessary that
  \begin{equation}
    \label{eq:asopg}
    \Pr[H^{(i)} \in A(i)] > 1 - \beta \text {  for all } i \in [h+
    2\gamma n + 1]
  \end{equation}

  Observe that $A(1)$ and $A(h+ 2\gamma n+1)$ are disjoint.
  Therefore, \autoref{eq:ahpgoiahsg} implies that
  $$\Pr[H^{(1,0)} \in A(1)] < \Pr[H^{(h + 2 \gamma n + 1, 0)} \in A(1)] + 1 -
  2 \beta$$ %
  By \autoref{eq:asopg} it follows that $\Pr[H^{(h + 2 \gamma n  + 1, 0)}
  \in A(1)] < \beta$ \Dnote{made previous less-than strict} and
  therefore from the previous inequality we 
  deduce that $\Pr[H^{(1,0)} \in A(1)] < 1 - \beta$.  But this
  contradicts \autoref{eq:asopg}, and therefore it must be the case
  that $M$ is not $([\eta + \gamma, \gamma], \beta)$-accurate.
\end{proof}

\begin{remark}
  A different way to evaluate the accuracy guarantee of our mechanism,
  (the one taken in the work of Ghosh and Roth \cite{GR11}) would be to
  compare it to the optimal accuracy achievable in the class of all
  envy-free mechanisms with budget $B$.  However, in our context it is
  not clear how to define envy-freeness: while it is clear what it
  means for player $i$ to receive player $j$'s payment, it is not at
  all clear (without making further assumptions) how to define the
  privacy loss of player $i$ as if he were treated like player $j$,
  since this loss may depend on the functional relationship between
  the player $i$'s type and the output of the mechanism.  Because of
  this, our mechanism may not envy-free (for reasonable privacy loss
  functions), and so we refrain from using envy-free mechanisms as a
  benchmark.
\end{remark}

\section{Acknowledgements}

\noindent K.N., S.V., and D.X., were supported in part by NSF grant
CNS-1237235, a gift from Google, Inc., and a Simons Investigator grant
to Salil Vadhan.  K.N. was also supported by ISF grant (276/12).
D.X. was also supported by the French ANR Blanc program under contract
ANR-12-BS02-005 (RDAM project).

\bibliographystyle{alphasy}
\bibliography{ref}



\end{document}